\documentclass[12pt]{article}
{\begin{Sbox}\begin{minipage}}%
{\end{minipage}\end{Sbox}\fbox{\TheSbox}}%

\usepackage[psamsfonts]{amssymb}
\usepackage{amsmath, amsthm, anysize, enumerate, color, url}
\usepackage{graphicx}
\usepackage{epstopdf}
\usepackage{epsfig}
\usepackage{subfigure}

\usepackage{algorithm}
\usepackage{algpseudocode}

\usepackage{natbib}
\usepackage{threeparttable}

\newtheorem{theorem}{Theorem} 
\newtheorem{lemma}{Lemma}[section] 
\newtheorem{proposition}{Proposition} 

\usepackage[colorlinks, breaklinks,
                   linkcolor=red,
                  anchorcolor=blue,
                  citecolor=blue
                  ]{hyperref}

\usepackage{breakurl}

\theoremstyle{definition}

\newcommand{\E}{\mathbb{E}}

\newcommand{\R}{\mathbb{R}}

\renewcommand{\P}{\mathbb{P}}

\newcommand{\bs}{\boldsymbol}

\usepackage{setspace}

\begin{document}

\title{A Probit Network Model with Arbitrary Dependence}

\author{Ting Yan\thanks{Department of Statistics, Central China Normal University, Wuhan, 430079, China.
\texttt{Email:} tingyanty@mail.ccnu.edu.cn.}
\\
Central China Normal University
}
\date{}

\maketitle

\begin{abstract}

In this paper, we adopt a latent variable method to formulate a network model with arbitrarily dependent structure.
We assume that the latent variables follow a multivariate normal distribution and a link between two nodes
forms if the sum of the corresponding node parameters exceeds the latent variable.
The dependent structure among edges is induced by the covariance matrix of the latent variables.
The marginal distribution of an edge is a probit function. We refer this model to as the \emph{Probit Network Model}.
We show that the moment estimator of the node parameter is consistent.
To the best of our knowledge, this is the first time to derive consistency result in a single observed network with globally dependent structures.
We extend the model to allow node covariate information.


\end{abstract}

\vskip 5 pt


\section{Introduction}
Networks or graphs provide a convenient representation to record complex interactive behaviors among a set of actors such as
friendships between peoples in social networks, paper citation in scientific literature, protein interaction in biology,
information diffusion between users in social media (e.g., Twitter, Sina Weibo) and so on.
With the development of information technology, the collected network datasets are increasing dramatically and large network datasets having millions of nodes are also
available nowadays. As a result, network data analyses have gained great interests not only from statisticians but also from
social scientists, mathematicians, physicists and computer scientists in recent years.
Many statistical approaches are proposed to analyze network data; see \cite{Robins2007173}, \cite{Goldenberg2010}, \cite{Fienberg2012} for some recent reviews.
The book by \cite{Kolaczyk2009} provides a comprehensive description on statistical analyses of network data.

One of basic measurements for an graph is its degrees. They preliminarily summarize the information contained in the graph and
many network models depend on the information of degrees directly or indirectly.
The degree heterogeneity is one important feature widely observed in realistic networks.
As an example in a well-known yeast dataset [\cite{yeast2002}] available at the R package ``igraphdata",
the node degree varies from the minimum value $1$ to the maximum value $118$ 
in its largest connected subgraph that has $2375$ nodes.
The $p_1$ model proposed by \cite{Holland:Leinhardt:1981} is the first one to model degree variation, in which
the bi-degrees of nodes and the number of reciprocated dyads form the sufficient statistics for the exponential distribution on directed graphs.

Since the work of \cite{Holland:Leinhardt:1981}, a class of network models parameterized by a set of node parameters are proposed,
including the chung-lu model [\cite{Chung:Lu:2002}] with the expected degrees as the parameters, the $\beta$-model [\cite{Chatterjee:Diaconis:Sly:2011,Blitzstein:Diaconis:2011,Park:Newman:2004}]
null models [\cite{Perry:Wolfe:2012}] and maximum entropy models [\cite{Hillar:Wibisono:2013}].
In these models, each node is assigned one parameter and therefore the number of parameters increases as the network size grows.
This leads to that asymptotic inference is nonstandard.
The asymptotic properties of the estimators (e.g., MLE or moment estimator) have been derived until recently [e.g., \cite{Chatterjee:Diaconis:Sly:2011,Rinaldo2013,Hillar:Wibisono:2013,Yan:Leng:Zhu:2016}].  In particular, \cite{Chatterjee:Diaconis:Sly:2011} prove the consistency of the maximum likelihood estimator (MLE) in the $\beta$-model
and \cite{Yan:Xu:2013} derive its asymptotic normal distribution.
\cite{Yan:Qin:Wang:2015} establish a unified theoretical framework for this class of models under which the consistency and asymptotical normality of the moment estimator hold.

In all above mentioned models, an assumption that all dyads are independent is imposed.
This is unrealistic since most of networks exhibit dependent structures such as transitivity and local clustering.
To admit dependence among edges and inherit good asymptotic properties of the estimator in dyad independent models, we propose a new model--\emph{Probit Network Model}. 
The model is implemented by using a latent variable method.
We assume that the latent variables are drawn from a multivariate normal distribution and a link between two nodes
forms if the sum of the corresponding two node parameters is no less than the latent variable.
The dependent structure among edges is induced by the covariance matrix of the latent variables.
We provide several cases through restricting the correlation parameters of the matrix to illustrate how the property of edge dependence
(e.g., $k$-stars, triangle) can be represented in our model.

An advantage of using the multivariate normal distribution is that the marginal distribution of a single random variable
does not depend on the correlated parameters of the covariance matrix and is also normal.
This leads to that the marginal probability distribution of any edge has a probit form.
Specifically, $P(a_{ij}=1)=\Phi(\mathbf{x}_{ij}^\top \beta )$, where $\Phi$ is the cumulative distribution function of
a standard normal distribution, $\beta$ is a vector of node parameters and $\mathbf{x}_{ij}$ is a vector with $i$th and $j$th elements $1$ and others $0$.
 and exactly a probit regression.
With this point, we call it \emph{Probit Network Model}.
Following the method developed in \cite{Yan:Qin:Wang:2015}, we show that the moment estimator of the node parameter is consistent.
To the best of our knowledge, this is the first time to derive consistency result in a single observed network with globally dependent structure.
This is sharply contrast with the existing results in dependent network models that some estimators such as the MLE is not consistent [\cite{Shalizi-Rinaldo2013}]
and the model is degeneracy [\cite{ChatterjeeDiaconis2013}].
We extend the model to admit node covariate information.

For the rest of the paper, we proceed as follows. In Section \ref{section:model}, we formally state the probit network model and
present the estimation of the model parameter.
In Section \ref{section:asymptotics}, we study the asymptotic property of the estimator.
In Section \ref{section:extension}, we extend the model to admit the node covariate information.
We make the summary and further discussion in Section \ref{section:discussion}.
All proofs are relegated into Appendix.

\section{Model and estimation}
\label{section:model}

Let $G_n$ be a simple undirected graph on $n\geq 2$ labelled nodes and
$A=(a_{ij})_{i,j=1}^n$ be its adjacency matrix.
Here, ``simple" means that $G_n$ does not have self-loops (i.e., $a_{ii}=0$) and
the edge between any two nodes is either present or absent.
The random variable $a_{ij}$ is equal to $1$ when there is an edge between $i$ and $j$; otherwise $a_{ij}=0$.
Let $d_i= \sum_{j \neq i} a_{ij}$ be the degree of node $i$
and $d=(d_1, \ldots, d_n)^\top$ be the degree sequence of the graph $G_n$.

\subsection{Model}
The probit network model is formulated as follows.
We assume that there exist a set of latent random variables $\{u_{ij}: i,j=1, \ldots, n; i<j\}$ and a set of unobserved node parameters $\{\alpha_i\}_{i=1, \ldots, n}$
such that an edge is present or absent according to the link surplus rule:
\[
a_{ij} = 1(\alpha_i + \alpha_j \ge u_{ij}),~~i<j,
\]
where $1()$ is an indicator function. Let $N=n(n-1)/2$.
Further, we assume that the $N$-dimensional vector $\mathbf{u}=(u_{12}, \ldots, u_{1n}, u_{23}, \ldots, u_{n-1, n})^\top$ with a lexicographical order
follows a multivariate normal distribution with mean $\mathbf{0}$ and covariance matrix $\Sigma$ whose diagonal elements are $1$'s
and row $i$ column $j$ nondiagonal  element is $\sigma_{ij}$ ($|\sigma_{ij}|<1$).
Since the marginal distribution for a single $u_{ij}$ is normal, we immediately have that the marginal distribution of $a_{ij}$ is
\[
\P(a_{ij}=1)=\Phi(\alpha_i + \alpha_j),
\]
where $\Phi$ is the cumulative distribution of the standard normal distribution.
It is noteworthy that
the use of the standard multivariate normal distribution causes no loss of generality compared with using an arbitrary mean and covariance matrix.
This is because adding a fixed amount to the mean can be compensated by adding the same amount from the intercept, and multiplying the standard deviation by a fixed amount can be compensated by multiplying the weight
in the additive term $(\alpha_i+\alpha_j)$ by the same amount.
The dependence among all edges is induced by the covariance structure of the normal vector $\mathbf{u}$.
Different structures of $\Sigma$ lead to different types of edge dependence of $G_n$.
An advantage of using the normal latent random variables is that the marginal distribution does not depend on the correlation coefficients of $\mathbf{u}$.

The dimension of the covariance matrix $\Sigma$ is very high, whose both numbers of rows and columns is in a magnitude of $n^2$.
Since $u_{ij}$ has double indices, it is not easy to see that which element of $\Sigma$ is $\mathrm{Cov}(u_{ij}, u_{kl})$.
So it is necessary to give the corresponding subscripts in $\sigma_{i_1, i_2}$ for the covariance $\mathrm{Cov}(u_{ij}, u_{kl})$ between $u_{ij}$ and $u_{kl}$.
Note that the subscripts in $u_{ij}$ satisfy $i<j$.
We assume that $i\le k$ since if $i>k$ we can interchange the positions of $u_{ij}$ and $u_{kl}$ in the expression $\mathrm{Cov}(u_{ij}, u_{kl})$.
We distinguish the first and second positions of $u_{ij}$ and $u_{kl}$ entering into the covariance formula
although the covariance is symmetric on its two variables.
If $i=k$, then we compare the second subscript $j$ and $l$. If $j\le l$, we use $\mathrm{Cov}(u_{ij}, u_{kl})$;
otherwise, we interchange $u_{ij}$ and $u_{kl}$ to ensure  $j \ge l$.
So we only consider the case $i\le k$ (if $i=k$, then $j\le l$).
By using a recursive method, one can show that the first subscript $\sigma_{i_1, i_2}$ for $\mathrm{Cov}(u_{ij}, u_{kl})$ is $[n-1+ \cdots + n - (i-1)]+j-i$ and the second is $[(n-i)+\cdots +n -k ]+l-k$.
We use the pair index $(ij)$ to represent the mathematical formula $[n-1+ \cdots + n - (i-1)]+j-i$. Using this notation we have, for
example, we write $\sigma_{(ij), (kl)}=\mathrm{Cov}(u_{ij}, u_{kl})$.

We say that a network has some subgraph property if the appearance probability of the subgraph is larger than the probability appearing in a random network.
Perfect subgraphs may be rare in real networks. As an illustration, perfect transitivity implies that, if $x$ is connected (through an edge) to $y$, and $y$ is connected to $z$, then $x$ is connected to $z$ as well.
In many social networks, the fact that $x$ knows $y$ and $y$ knows $z$ does not guarantee that $x$ knows $z$ as well, but makes it much more likely. Also the friend of my friend is not necessarily my friend, but is far more likely to be my friend than some randomly chosen member of the population. So in practice we are interested in subgraph property not the subgraph itself.
First, we look at the simplest subgraph structure--$2$-star that is a subset
of three nodes in which one node is connected by a link to each of the other two.
We consider a node set $\{i,j,k\}$. When $u_{ij}$ and $u_{ik}$ becomes smaller simultaneously, the edges $(i,j)$ and $(i,k)$ forms more likely
such that the star in this set with $i$ as the center node appears probably.
It leads to the condition that $\sigma_{(ij)(ik)}$ is positively correlated.
In other words, if a large proportion of correlation coefficients are positive, then such covariance structure induces the $2$-star property.
Next, we consider a little complex subgraph--triangle that is
a subset of three mutually connected nodes.
To represent the transitivity structure, we can make the restriction: the triad $(\sigma_{(ij)(jk)}, \sigma_{(jk)(ki)}, \sigma_{(ki)(jk)})$ is positive.
A more strong restriction is that other than having the same sign each value exceeds a threshold.
For other subgraphs, we can use a similar manner to construct the covariance matrix.

\subsection{Estimation}
If we don't make any restriction on the covariance matrix $\Sigma$, the total number of parameters is
$n+n(n-1)/2$ in the PNM since the number of node parameters is $n$ and the number of correlation parameters is $n(n-1)/2$.
Note that there are only $n(n-1)/2$ observed random variables $\{a_{ij}\}_{i<j}$ in the graph $\mathcal{G}_n$.
Therefore, the number of parameters is larger than the number of random variables such that
some parameters can not be estimated.
One method is to use the penalized likelihood method (e.g., $L_1$ penalized function) to obtain the sparse solution.
However, the optimization problem is involved with a very high-dimensional integrate such that the computation is not feasible.
Another method is reducing the number of correlation parameters. We use this method here and focus two cases for the covariance matrix.
The first one is simple to set all correlation coefficients sharing one common parameter $\sigma_0$ and letting
$\sigma_{(ij)(kl)} = \sigma_0^{|(ij)-(kl)|} (=\mathrm{Cov}(u_{ij}, u_{kl}))$. The second one is
\[
\sigma_{(ij)(kl)}= \sigma_0 + \sigma_i + \sigma_j + \sigma_k + \sigma_l,
\]
with the restriction $\sum_i \sigma_i=0$, where $\sigma_0$ is the intercept term.

To estimate the model parameters, we utilize a two-stage method.
First, we use the moment equations based on the node degrees to estimate the node parameters $\{\alpha_i\}_{i=1}^n$:
\begin{equation}\label{eq:moment}
d_i = \sum_{j\neq i}^n \Phi( \alpha_i + \alpha_j ), ~~ i=1, \ldots, n.
\end{equation}
The solution to the above equation (i.e., the moment estimator) is denoted by $\widehat{\alpha}$.
In the second stage, we estimate the correlation parameters also based on the moment equations by combining the random variables $\{a_{ij}\}_{i<j}$, where
the estimates $\widehat{\alpha}$ from the first stage are plugged.
For example, in the first case that all $\sigma_{ij}$'s share one common parameter $\sigma_0$, we can use the moment equation:
\[
\sum_{i<j, k<l} [a_{ij}a_{kl} - \E (a_{ij}a_{kl}) ] = 0,
\]
where $\E [a_{ij}a_{kl}]$ is the probability of both $a_{ij}$ and $a_{kl}$ equalling to $1$:
\[
\begin{array}{lcl}
\P(a_{ij}=1, a_{ji}=1) & = & \int_{-\infty}^{\alpha_i+\alpha_j} \int_{-\infty}^{\alpha_k+\alpha_l} \phi(t_1, t_2, \sigma_0) dt_2 dt_1.
\end{array}
\]
In the above equation, $\phi$ is the standardized bivariate density function:
\[
\phi(t_1, t_2, \sigma_0) = \frac{1}{ 2\pi \sqrt{1-\sigma_0^2} } \exp \left (  -\frac{ t_1^2 +t_2^2 - 2\sigma_0 t_1 t_2}{ 2(1-\sigma_0^2) } \right).
\]
In the other case, we can use the $n$ moment equations:
\[
\sum_{j\neq l} [a_{ij}a_{il} - \E (a_{ij}a_{il}) ] = 0, ~~i=1, \ldots, n.
\]

\section{Asymptotical properties}
\label{section:asymptotics}
In this section, we establish the consistency of the moment estimator for the node parameters.
We first give some notations. For a subset $C\subset \R^n$, let $C^0$ and $\overline{C}$ denote the interior and closure of $C$, respectively.  For a vector $x=(x_1, \ldots, x_n)^\top\in R^n$, denote by
$\|x\|_\infty = \max_{1\le i\le n} |x_i|$, the $\ell_\infty$-norm of $x$.  For an $n\times n$ matrix $J=(J_{i,j})$, let $\|J\|_\infty$ denote the matrix norm induced by the $\ell_\infty$-norm on vectors in $\R^n$, i.e.
\[
\|J\|_\infty = \max_{x\neq 0} \frac{ \|Jx\|_\infty }{\|x\|_\infty}
=\max_{1\le i\le n}\sum_{j=1}^n |J_{i,j}|.
\]

Note that $\E(a_{ij})$ only depends on the sum $\alpha_i + \alpha_j$ and $\Phi(x)$ is the cumulative distribution function of the standard
normal random variable.
Define a system of functions:
\begin{equation}\label{eq:F-definition}
\begin{array}{rcl}
F_i(\bs{\alpha}) & = & d_i - \E(d_i) =d_i - \sum_{j=1; j\neq i}^n \Phi(\alpha_i + \alpha_j),~~i=1,\ldots, n, \\
F(\bs{\alpha}) & = & (F_1(\bs{\alpha}), \ldots, F_n(\bs{\alpha}))^\top. \\
\end{array}
\end{equation}
It is clear that the solution to $F(\bs{\alpha})=0$ is the moment estimator $\bs{\widehat{\alpha}}$ of $\bs{\alpha}$.
The idea to establish the consistency result is to obtain a geometrically fast convergence of rate for the Newton iterative sequence constructed from
$F(\bs{\alpha})=0$. This method is used in \cite{Yan:Qin:Wang:2015}.
One key point is to derive the upper bound of the maximum value of the centered $d_i$'s.
Generally, it is difficult to handle the dependent random variables. Since $\mathbf{u}$ is a normal random vector,
there are two special cases for the covariance matrix $\Sigma$ in which $\mathbf{u}$ exhibits a nice dependent structure.
That is, when all $\sigma_{ij}>0$, $\mathbf{u}$ is positively dependent random variables [\cite{pitt-1982}];
when all $\sigma_{ij}<0$,  $\mathbf{u}$ is negatively dependent random variables [\cite{joag-dev1983}].
Under these two cases, there exist exponential probability inequalities for the sequence $\{a_{ij}\}_{j=1, j\neq i}^n$ for each fixed $i$
that is used to derive the upper bound of the maximum value of the centered $d_i$'s.
Formally, we present the consistency of the moment estimator $\widehat{\alpha}$ as follows.

\begin{theorem}\label{Theorem-con}
(1) Assume that (i) $\sigma_{ij}\le 0$ for all $i\neq j$ or (ii)  $\sigma_{ij}\ge 0$ for all $i\neq j$ and
$\max_{i\neq j \neq k} \sigma_{(ij)(ik)} \le  \exp( -\frac{8}{3}(r\log n)^{1/2})=O( e^{-n^{1/2}})$, where $r=O(n/\log n)$.
Let $\alpha^*$ denote the true parameter.
If $\| \alpha^* \|\le Q_n$ and $e^{3Q_n^2}=o(n^{1/2})$, then we have
\begin{equation}\label{Newton-convergence-rate}
\| \widehat{\alpha} - \alpha^* \|_\infty = O_p( n^{-1/2}e^{3Q_n^2})=o_p(1).
\end{equation}
\end{theorem}

\section{Extensions}
\label{section:extension}
The extension of the PNM to directed networks is trivial.
The only difference is that for each node $i$ there is one out-degree parameter $\alpha_i$ and one in-degree parameter $\beta_j$
and the latent variable $\mathbf{u}$ is $n(n-1)$ dimensional.
So the link surplus rule is $a_{ij}=1(\alpha_i + \beta_j \ge u_{ij} )$.

The link surplus rule is also easily extended into allow the covariate information of nodes.
Assume that we observe a vector $Z_{ij}$, the covariate information attached to the edge between nodes $i$ and $j$.
The edge covariate can be constructed according to exogenous information such as node attributes.
Given the attributes $X_i$ and $X_j$ of nodes $i$ and $j$,
the covariate $Z_{ij}$ can be formed according to the similarity or dissimilarity between node attributes $X_i$ and $X_j$.
Specifically, $Z_{ij}$ can be represented through a symmetric function $g(\cdot, \cdot)$ with $Z_i$ and $Z_j$ as its arguments. As an example
if $X_{i1}$ and $X_{i2}$ are location coordinates, then $Z_{ij} =[(X_{i1}- X_{j1})^2 + (X_{i2} - X_{j2})^2]^{1/2}$, denoting the Euclidean distance
between $i$ and $j$. So the link surplus rule is
\[
a_{ij}=1(Z_{ij}^\top \gamma + \alpha_i + \alpha_j \ge u_{ij} ).
\]
Such rule captures the homophily--the tendency of individuals to associate and bond with similar others.
\cite{Graham2017} also use the type of the above rule with the assumption of independent edges and $u_{ij}$ as the logistic distribution.
\cite{Dzemski2017} study the directed version of the above link rule with the assumption of independent dyad edges using a scalar parameter to characterize the correlation of dyads
in which a two-step approach was used for estimation and the focus is on the homophily parameter.

\section{Discussion}
\label{section:discussion}
We have proposed the probit network model to admit the non-independence of edges.
The structure of dependence with graphs are determined by the covariance of the normal latent random variables.
The proposed model addresses the shortcoming of the dyad independent model and inherits its nice asymptotic properties.
However, many open problems remain, including the issue of the asymptotic distribution of the moment estimators for node parameters as well as
the asymptotic behaviors of the estimator for the correlation parameters.
Despite these challenges, we believe that the probit network model provides a good starting point to analyze complex network data.

\section{Appendix}
\subsection{Preliminaries}

We first give the definition of the terminologies ``positively associated" [\cite{esary1967}] (in the original terminology, associated) and ``negatively associated" [\cite{joag-dev1983,block1982,Ghosh1981}].
A finite family of random variables $\{X_i, 1 \le i \le n\}$ is said to be positively associated (PA) if for
every pair of disjoint subsets $A$ and $B$ of $\{1, 2, \ldots, n\}$,
\[
\mathrm{Cov}[ f(X_i, i \in A), g(X_j, j \in B) ] \ge 0,
\]
whenever $f$ and $g$ are coordinatewise increasing and the covariance exists. An infinite family is
positively associated if every finite subfamily is positively associated.
Analogously, if
\[
\mathrm{Cov}[ f(X_i, i \in A), g(X_j, j \in B) ] \le 0,
\]
then $\{X_i, 1 \le i \le n\}$ is said to be negatively associated (NA).
An infinite family is
negatively associated if every finite subfamily is negatively associated.
Negative/Positive association has one nice property that
nondecreasing functions of disjoint sets of NA/PA random variables are also NA/PA [\cite{esary1967,joag-dev1983}].
This type of closure property does not hold for some other types of negative dependence defined in \cite{lehmann1966} (see \cite{joag-dev1983}).
For a set of normal random variables, \cite{pitt-1982} and \cite{joag-dev-pitt1983} establish the necessary and sufficient condition for them to be PA;
\cite{joag-dev1983} derive the corresponding condition for them to be NA. For clarify, we restate them below.

\begin{theorem}\label{theorem-PANA}
Let $X=(X_1, \ldots, X_n)$ be multivariate normal with mean vector $0$ and covariance matrix
$\Sigma=(\sigma_{ij}=\mathrm{Cov}[X_i, X_j])$. \\
(i)(\cite{pitt-1982,joag-dev-pitt1983}) The condition
$\sigma_{ij}\ge 0$ for $1\le i,j \le k$ is necessary and sufficient for the random vector $X$
to be positively associated. \\
(ii)(\cite{joag-dev1983}) The condition
$\sigma_{ij}\le 0$ for $1\le i,j \le k$ is necessary and sufficient for  the random vector $X$
to be negatively associated.
\end{theorem}

Let $S_n=\sum_{i=1}^n X_i$. Exponential bounds for the probabilities $P(|S_n| \ge \epsilon)$ ($\epsilon>0$)
play an important role for the purpose of providing rates of convergence
for estimates of various quantities. A well-known exponential bounds is the \citeauthor{Hoeffding:1963} inequality for independent bounded
random variables. As pointed by \cite{Roussas1996}, the proof of the Hoeffding inequality goes through almost
unchanged, with one small change with one equality is replaced by the inequality that automatically holds for NA random variables.
We state it below.

\begin{theorem}(Proposition 3.1 in \cite{Roussas1996})\label{theorem:hoeffing:NA}
Let $X_1, \ldots, X_n$ be negatively associated random variables such that $a_i \le X_i \le b_i$, $i=1, \ldots, n$.
Set $S_n = \sum_{i=1}^n X_i$. Then, for every $t>0$,
\[
\P( | S_n - \E S_n | \ge t ) \le 2 \exp [ -2t^2/\sum_{i=1}^n (b_i-a_i)^2]
\]
\end{theorem}

For PA random variables, \cite{Ioannides1999423} prove that under some additional conditions the Hoeffding inequality hold.

\begin{theorem}(\cite{Ioannides1999423})\label{theorem:hoeffing:PA}
Assume that $X_1, \ldots, X_n$ are PA and bounded, $|X_i| \le M/2$, $i\ge 1$.
For positive integers
$1\le p = p(n) <n$ and $p\to\infty$, divide the set $\{1, 2, \ldots, n\}$ into successive groups each containing
$p$ elements. Let $r=r(n)$ be the largest integer with:
$0<r<n$, $r\to\infty$ and $2pr \le n$.
Let $C(k)=\sup \{ \mathrm{Cov}(X_i, X_{i+k}): i \ge 1\}$, $k\ge 1$, and assume that $C(k)$ is nonincreasing as $k\to\infty$. Let $\alpha$ be an arbitrary constant $>1$.
Let $S_n=\sum_{i=1}^n (X_i - \E X_i)$ and $\epsilon_n=(\alpha M^2/2)^{1/2} ( \log n / r)^{1/2}$.
If
\begin{equation}\label{inequality-cp}
C(p) \le \exp[ - \frac{ 4(M+1) }{3M} (\frac{\alpha}{2})^{1/2} ( r\log n)^{1/2} ]
\end{equation}
holds, then for sufficiently large $n$
\[
\P( |\frac{1}{n} S_n| \ge \epsilon_n ) \le C_0 \exp( -\frac{2r\epsilon_n^2}{9M^2}  ),
\]
where $C_0$ is a constant (for example, $C_0=12$)
\end{theorem}

Other than the above preliminary results, we need two more results in order to start the proof.
Given $m, M>0$, we say an $n\times n$ matrix $V=(v_{ij})$ belongs to the matrix class $\mathcal{L}_{n}(m, M)$ if
$V$ is a diagonally balanced matrix with positive elements bounded by $m$ and $M$,
\begin{equation}\label{eq-V}
\begin{array}{l}
v_{ii}=\sum_{j=1, j\neq i}^{n} v_{ij}, ~~i=1,\ldots, n, \\
m\le v_{ij} \le M, ~~ i,j=1,\ldots,n; i\neq j.
\end{array}
\end{equation}
The first is on the approximation error of using $S_n=\mathrm{diag}(1/v_{11}, \ldots, 1/v_{nn})$ to approximate the inverse of $V_n=(v_{ij})_{n\times n}$ belonging to the matrix class $\mathcal{L}_n(m_n, M_n)$.
\cite{Yan:Zhao:Qin:2015} obtain the upper bound of the approximation error stated below,
which has an order $n^{-2}$.

\begin{proposition}[Proposition 1 in \cite{Yan:Zhao:Qin:2015}] \label{pro:inverse:appro}
If $V_n\in \mathcal{L}_n(m_n, M_n)$, then for $n\ge 3$, the following holds:
\begin{equation}\label{O-upperbound}
\|V^{-1}_n - S_n \|
 \le  \frac{ M_n(nM_n+(n-2)m_n)}{2m_n^3(n-2)(n-1)^2}+\frac{1}{2m_n(n-1)^2}+\frac{1}{m_n n(n-1)}=O(\frac{M_n^2}{n^2m_n^3}),
\end{equation}
where $\|A\|$ is the matrix maximum norm $\|A\|:=\max_{i,j} |a_{ij}|$ for a general matrix $A$.

\end{proposition}

From the above proposition, according the definition of $\|\cdot\|_\infty$-norm, we immediately have:
\begin{lemma}\label{lemma:inverse:bound}
Assume that $V\in\mathcal{L}_{n}(m, M)$. For large enough $n$,
\[
\| V^{-1}_n \|_\infty \le \| V^{-1}_n-S_n \|_\infty + \|S_n\|_\infty = O(\frac{M_n^2}{nm_n^3}).
\]
\end{lemma}

The other result is on the rate of convergence for the Newton's method.
There are many convergence results on the Newton's method; see the book \cite{suli:Mayers:2003} for a comprehensive survey.
We use \citeauthor{Gragg:Tapia:1974}'s \citeyearpar{Gragg:Tapia:1974} result here.
Let $F(x): \R^n \to \R^n$. We say that a Jacobian matrix $F^\prime(x)$ with $x\in \R^n$ is lipschitz continuous on a convex set $D\subset\R^n$ if
for any $x,y\in D$, there exists a constant $\lambda>0$ such that
for any vector $v\in \R^n$ the inequality
\[
\| F^\prime (x) (v) - F^\prime (y) (v) \| \le \lambda \| x - y \| \|v\|
\]
holds.

\begin{theorem}[\cite{Gragg:Tapia:1974}]\label{pro:Newton:Kantovorich}
Let $D$ be an open convex set of $\R^n$ and $F:D \to \R^n$ a differential function
with a Jacobian $F^\prime(x)$ that is Lipschitz continuous on $D$ with Lipschitz parameter $\lambda$.
Assume that $x_0 \in D$ is such that $[ F^\prime (x_0) ]^{-1} $ exists,
\[
\| [ F^\prime (x_0 ) ]^{-1} \|_\infty  \le \aleph,~~ \| [ F^\prime (x_0) ]^{-1} F(x_0) \| \le \delta, ~~ \rho= 2 \aleph \lambda \delta \le 1,
\]
and
\[
S(x_0, t^*) \subset D, ~~ t^* = \frac{2}{\rho} ( 1 - \sqrt{1-\rho} ) \delta = \frac{ 2\delta }{ 1 + \sqrt{1+\rho} }.
\]
Then: (1) The Newton iterates $x_{n+1} = x_n - [ F^\prime (x_n) ]^{-1} F(x_n)$ exist and $x_n \in S(x_0, t^*) \subset D$ for $n \ge 0$. (2)
$x^* = \lim x_n$ exists, $x^* \in \overline{ S(x_0, t^*) } \subset D$ and $F(x^*)=0$.
\end{theorem}

\subsection{Proofs for Theorem \ref{Theorem-con}}

\begin{lemma}\label{lemma-theorem-con}
(1) Assume that $\sigma_{ij}\le 0$ for all $i\neq j$.  With probability at least $1-2n/(n-1)^2$, we have
\[
\max_{1\le i\le n} |d_i - \E(d_i) | \le  (n\log n)^{1/2}.
\]
(2) Assume that $\sigma_{ij}\ge 0$ for all $i\neq j$. Let $r=O(n/\log n)$. If
$\max_{i\neq j \neq k} \rho_{(ij)(ik)} \le  \exp( -\frac{8}{3}(r\log n)^{1/2})=O( e^{-n^{1/2}})$,
then with probability at least $1-2/n$, we have
\[
\max_{1\le i\le n} |d_i - \E(d_i) | = O( n^{1/2}\log n).
\]
\end{lemma}

\begin{proof}
We first prove the first part. Since $\sigma_{ij}\le 0$ for all $i\neq j$,
by Theorem \ref{theorem-PANA} (ii), the latent random variables $(u_{12}, u_{13}, \ldots, u_{n-1, n})$ are NA.
By the property of NA, for any fixed $i$, $a_{ij}$'s ($j=1, \ldots, n; j\neq i$) are NA.
Note that $a_{ij}$ is an indictor random variable and $d_i$ is a sum of $n-1$ $NA$ random variables.
By Theorem \ref{theorem:hoeffing:NA}, we have
\[
\P( | d_i - \E d_i | \ge \sqrt{n\log n} ) \le 2 \exp( - 2\cdot \frac{n\log n}{n-1} ) \le \frac{2}{n^2}.
\]
Therefore,
\[
P\left( \max_i|d_i - \E d_i | \ge \sqrt{n\log n} \right) \le  n\times \frac{2}{n^2}=\frac{2}{n}.
\]
This is equivalent to the first part.

Now we apply Theorem \ref{theorem:hoeffing:PA} to prove the second part.
When $\rho_{(ij)(ik)}\ge 0$, we have
\[
0\le Cov( a_{ij}, a_{ik} ) = P( u_{ij}<\alpha_i+\alpha_j, u_{ik}<\alpha_i+\alpha_k)-P(u_{ij}<\alpha_i+\alpha_j)P(u_{ik}<\alpha_i+\alpha_k) < \sigma_{(ij)(ik)}.
\]
This shows
\[
C(k)=\sup\{ Cov(a_{ij}, a_{i,j+k}): j\ge 1, j\neq i\} \le \max_{j,k\neq i, j\neq k}\sigma_{(ij)(ik)}.
\]
In Theorem \ref{theorem:hoeffing:PA}, we set $\alpha=2$ and $r=O(n/\log n)$.
Since $a_{ij}$ is an indictor random variable, $M=1$. So $\epsilon_n$ in Theorem \ref{theorem:hoeffing:PA} is equal to $(\log n/r)^{1/2}$ and $p=O(\log n)$.
Note that
\[
C(p) < \max_{i\neq j \neq k} \sigma_{(ij)(ik)} \le  \exp( -\frac{8}{3}(r\log n)^{1/2})=O( e^{-n^{1/2}}).
\]
It demonstrate that condition \eqref{inequality-cp} holds.
By Theorem \ref{theorem:hoeffing:PA} in which we choose $\epsilon_n=3(\log n/r)^{1/2}$, it yields
\[
\P(|d_i - \E d_i | \ge  3(n\log n)^{1/2} (n/r)^{1/2} ) \le C_0 \exp( - 2 \log n )= \frac{C_0}{n}.
\]
The left proof is similar to the proof of the first part and we omit it.
\end{proof}

Recall the definition $F(\bs{\alpha})$ in \eqref{eq:F-definition}.
Let $F'(\bs{\alpha})$ denote the
Jacobian matrix of $F(\bs{\alpha})$ on $\bs{\alpha}$.
We consider the parameter space:
\[
\Theta = \{ \theta : - Q_n \le \alpha_i + \alpha_j \le Q_n, 1\le i, j \le n \}.
\]
Let $\phi(x)$ be the standard normal density function. The elements of
the Jacobian matrix $F'(\alpha)$ of $F(\alpha)$ are: for $i=1,\ldots, n$,
\[
\frac{\partial F_i }{\partial \alpha_i} =  \sum_{j\neq i} \phi( \alpha_i+\beta_j),~~
~~\frac{\partial F_i}{\partial \alpha_j}=  \phi( \alpha_i+\beta_j), i\neq j.
\]
Since $\phi(x)=(2\pi)^{1/2}e^{-x^2/2}$ is an decreasing function on $|x|$, we have when $|x| \le Q_n$,
\[
 \frac{1}{2\pi} e^{ -Q_n^2/2} \le \phi(x) \le \frac{1}{2\pi}.
\]
So $-F^\prime(\alpha)$ belongs to $\mathcal{L}(m_n, M_n)$, where $m_n=(2\pi)^{-1} e^{ -Q_n^2/2}$ and~¡¡$M_n=(2\pi)^{-1}$.

\begin{lemma}\label{lemma:jacobian}
The Jacobian matrix $F'( \alpha )$ of $F(\alpha)$ on $\alpha$ is Lipschitz continuous with the Lipschitz parameter $(n-1)(2/(e\pi))^{1/2}$.
\end{lemma}

\begin{proof}
Let $\mathbf{x}, \mathbf{y} \in R^{n}$ and
\[
F_i'(\bs{\alpha}) = (F_{i,1}'(\bs{\alpha}), \ldots, F_{i,n}'(\bs{\alpha})).
\]
Then, for $i=1, \ldots, n$, we have
\[
\frac{\partial^2 F_i }{\partial \alpha_s\partial \alpha_l }=0, ~s\neq l;
~~~
\frac{\partial^2 F_i }{\partial \alpha_i^2 }=  -  \sum_{k \neq i} \frac{\alpha_i +\alpha_k}{\sqrt{2\pi}} e^{-(\alpha_i+\alpha_k)^2/2},
\]
By the mean value theorem for vector-valued functions (\citeauthor{Lang1993}, \citeyear{Lang1993}, p.341), we have
\[
F'_i(\mathbf{x}) - F'_i(\mathbf{y}) = J^{(i)}(\mathbf{x}-\mathbf{y}),
\]
where
\[
J^{(i)}_{s,l} = \int_0^1 \frac{ \partial F'_{i,s} }{\partial \alpha_l}(t\mathbf{x}+(1-t)\mathbf{y})dt,~~ s,l=1,\ldots, n.
\]
Let $h(x)=xe^{-x^2/2}$. Then $h^\prime (x)= (1-x^2)e^{-x^2/2}$.
Therefore, when $x\in (0,1)$, $h(x)$ is an increasing function on its argument $x$; when $x\in (1, \infty)$,
$h(x)$ is an decreasing function on $x$. As a result, $h(x)$ attains its maximum value at $x=1$ when $x>0$.
Since $h(x)$ is a symmetric function,  we have $|h(x)|\le e^{-1/2}\approx 0.6$.
Therefore,
\[
\max_s \sum_{l=1}^{n} |J^{(i)}_{s,l}|\le (n-1)\frac{1}{\sqrt{2e\pi}}, ~~\sum_{s,l}|J^{(i)}_{s,l}|\le (n-1)\sqrt{\frac{2}{e\pi}}.
\]
Consequently,
\[
\| F_i'(\mathbf{x}) - F_i'(\mathbf{y}) \|_\infty \le \| J^{(i)} \|_\infty \| \mathbf{x} - \mathbf{y}\|_\infty \le (n-1)\frac{1}{\sqrt{2e\pi}},~~~i=1,\ldots, n.
\]
Then for any vector $\mathbf{v}\in R^{n}$,
\begin{eqnarray*}
\| [F'(\mathbf{x}) - F'(\mathbf{y})]\mathbf{v} \|_\infty  & = & \max_i |\sum_{j=1}^{n} ( F_{i,j}'(\mathbf{x}) - F_{i,j}'(\mathbf{y}) ) v_j | \\
& = & \max_i |(\mathbf{x}-\mathbf{y})J^{(i)} \mathbf{v} | \\
& \le & \|\mathbf{x}-\mathbf{y}\|_\infty \| \mathbf{v}\|_\infty \sum_{k,l}|J^{(i)}_{k,l}|\le (n-1)\sqrt{\frac{2}{e\pi}}
\|\mathbf{x}-\mathbf{y}\|_\infty \| \mathbf{v}\|_\infty.
\end{eqnarray*}
It shows that $F'(\mathbf{x})$ is Lipschitz continuous with the Lipschitz parameter $(n-1)(2/(e\pi))^{1/2}$.

\end{proof}

\begin{proof}[Proof of Theorem \ref{Theorem-con}]
The proofs of the first and second parts in Theorem \ref{Theorem-con} are similar and both apply the Newton-Kantovorich theorem.
We only give the proof of the first part and omit the proof of the second part.

To prove the first part, it is sufficient to show that the Newton-Kantovorich conditions hold.
In the Newton's iterative step, we take the true parameter vector $\bs{\alpha}$
as the starting point $\bs{\alpha}^{(0)}:=\bs{\alpha}$.
Note that when $\bs{\alpha} \in \Theta$, $-F'(\bs{\alpha}^*)\in \mathcal{L}_{2n-1}(m_*, M_*)$, where
\[
M_*=\frac{1}{\sqrt{2\pi e}}, ~~ m_*= \frac{1}{\sqrt{2\pi}}  e^{- Q_n^2/2 }.
\]
Let $V=F'(\bs{\alpha})\in \mathcal{L}_n(m_n, M_n)$ and $W= V^{-1} - S$.
By Lemma \ref{lemma:inverse:bound}, we have $\aleph = c_2M^2/(nm^3)$.
Note that $F(\bs{\alpha}) = \mathbf{d} - \E (\mathbf{d})$.
Lemma \ref{lemma:jacobian} shows that $F'(\bs{\alpha})$ is Lipschitz continuous. 
Under the condition that $\sigma_{ij}\le 0$ or $\sigma_{ij}\ge 0$ and $\max_{i\neq j \neq k} \rho_{(ij)(ik)} \le  \exp( -\frac{8}{3}(r\log n)^{1/2})=O( e^{-n^{1/2}})$,  Lemma \ref{lemma-theorem-con} shows that
\[
\| F(\bs{\alpha}) \| = O_p( n^{1/2}\log n). 
\] 
By Lemma \ref{lemma:inverse:bound} and Proposition \ref{pro:inverse:appro}, we have
\begin{eqnarray*}
\| [F'(\bs{\alpha})]^{-1}F(\bs{\alpha}) \|_\infty & \le &
n\|W\| \|F(\bs{\alpha})\|_\infty + \max_{i}\frac{|F_i(\bs{\alpha})|}{v_{ii}}
 \\
& \le & O \left(\frac{M_n^2}{ nm_n^3}  \right) \times O_p( n^{1/2} \log n ) \\
& = &  O_p( n^{-1/2} e^{3Q_n^2/2} \log n ). 
\end{eqnarray*}
Therefore, we can choose
\[
\delta = O( n^{1/2} e^{3Q_n^2/2} \log n ),
\]
such that 
\[
h=2\aleph \lambda \delta =  \frac{ e^{3Q_n^2/2} }{n} \times O(n) \times O( \frac{e^{3Q_n^2/2}\log n}{ n^{1/2} }) .
\]
By Theorem \ref{pro:Newton:Kantovorich}, if $e^{3Q_n^2}=o(n^{1/2})$,  then
$\|\bs{\widehat{\alpha}} - \bs{\alpha} \|_\infty = o_p(1)$.
\end{proof}

\section*{Acknowledgements}
Yan is partially supported by
by the National Natural Science Foundation of China (No. 11771171).

\bibliography{reference2}

\bibliographystyle{apa}
\addtolength{\itemsep}{-2 em}

\end{document}